\documentclass[conference]{IEEEtran}
\usepackage[margin=0.75in]{geometry}
\usepackage{enumerate}
\usepackage{amsmath,amsthm,amssymb}
\usepackage{framed}
\usepackage{graphicx}
\usepackage[noadjust]{cite}
\usepackage{epstopdf}
\usepackage{caption}
\usepackage{color}
\usepackage{flushend}
\usepackage{empheq}
\usepackage{soul}
\usepackage{subcaption}
\usepackage{cite}
\usepackage[ruled,vlined]{algorithm2e} 
\newcommand{\R}{\mathbb{R}}
\newcommand{\N}{\mathbb{N}}

\newcommand{\un}[1]{\underline{#1}}
\newcommand{\ov}[1]{\overline{#1}}
\newcommand{\dist}{\operatorname{dist}}

\IEEEoverridecommandlockouts

\newboolean{showcomments}
\setboolean{showcomments}{true}
\newcommand{\hao}[1]{  \ifthenelse{\boolean{showcomments}}
{ \textcolor{red}{(CE says:  #1)}} {}  }
\newcommand{\lina}[1]{\ifthenelse{\boolean{showcomments}}
{ \textcolor{blue}{(Lina says:  #1)}} {}  }

\begin{document}

\newtheorem{defin}{Definition}
\newtheorem{theorem}{Theorem}
\newtheorem{prop}{Proposition}
\newtheorem{lemma}{Lemma}
\newtheorem{corollary}{Corollary}
\newtheorem{alg}{Algorithm}
\newtheorem{remark}{Remark}
\newtheorem{notations}{Notations}
\newtheorem{assumption}{Assumption}
\newtheorem{operation}{Operation}
\newtheorem{example}{Example}
\newtheorem{question}{Question}

\newcommand{\be}{\begin{equation}}
\newcommand{\ee}{\end{equation}}
\newcommand{\ba}{\begin{array}}
\newcommand{\ea}{\end{array}}
\newcommand{\bea}{\begin{eqnarray}}
\newcommand{\eea}{\end{eqnarray}}
\newcommand{\combin}[2]{\ensuremath{ \left( \ba{c} #1 \\ #2 \ea \right) }}
\newcommand{\diag}{{\mbox{diag}}}
\newcommand{\rank}{{\mbox{rank}}}
\newcommand{\dom}{{\mbox{dom{\color{white!100!black}.}}}}
\newcommand{\range}{{\mbox{range{\color{white!100!black}.}}}}
\newcommand{\image}{{\mbox{image{\color{white!100!black}.}}}}
\newcommand{\herm}{^{\mbox{\scriptsize H}}}  
\newcommand{\sherm}{^{\mbox{\tiny H}}}       
\newcommand{\tran}{^{\mbox{\scriptsize T}}}  
\newcommand{\tranIn}{^{\mbox{-\scriptsize T}}}  
\newcommand{\card}{{\mbox{\textbf{card}}}}
\newcommand{\asign}{{\mbox{$\colon\hspace{-2mm}=\hspace{1mm}$}}}
\newcommand{\ssum}[1]{\mathop{ \textstyle{\sum}}_{#1}}

\newcommand{\vbar}{\raisebox{.17ex}{\rule{.04em}{1.35ex}}}
\newcommand{\vbarind}{\raisebox{.01ex}{\rule{.04em}{1.1ex}}}
\newcommand{\D}{\ifmmode {\rm I}\hspace{-.2em}{\rm D} \else ${\rm I}\hspace{-.2em}{\rm D}$ \fi}
\newcommand{\T}{\ifmmode {\rm I}\hspace{-.2em}{\rm T} \else ${\rm I}\hspace{-.2em}{\rm T}$ \fi}
\newcommand{\B}{\ifmmode {\rm I}\hspace{-.2em}{\rm B} \else \mbox{${\rm I}\hspace{-.2em}{\rm B}$} \fi}
\newcommand{\Hil}{\ifmmode {\rm I}\hspace{-.2em}{\rm H} \else \mbox{${\rm I}\hspace{-.2em}{\rm H}$} \fi}
\newcommand{\C}{\ifmmode \hspace{.2em}\vbar\hspace{-.31em}{\rm C} \else \mbox{$\hspace{.2em}\vbar\hspace{-.31em}{\rm C}$} \fi}
\newcommand{\Cind}{\ifmmode \hspace{.2em}\vbarind\hspace{-.25em}{\rm C} \else \mbox{$\hspace{.2em}\vbarind\hspace{-.25em}{\rm C}$} \fi}
\newcommand{\Q}{\ifmmode \hspace{.2em}\vbar\hspace{-.31em}{\rm Q} \else \mbox{$\hspace{.2em}\vbar\hspace{-.31em}{\rm Q}$} \fi}
\newcommand{\Z}{\ifmmode {\rm Z}\hspace{-.28em}{\rm Z} \else ${\rm Z}\hspace{-.38em}{\rm Z}$ \fi}

\newcommand{\sgn}{\mbox {sgn}}
\newcommand{\var}{\mbox {var}}
\newcommand{\E}{\mbox {E}}
\newcommand{\cov}{\mbox {cov}}
\renewcommand{\Re}{\mbox {Re}}
\renewcommand{\Im}{\mbox {Im}}
\newcommand{\cum}{\mbox {cum}}

\renewcommand{\vec}[1]{{\bf{#1}}}     
\newcommand{\vecsc}[1]{\mbox {\boldmath \scriptsize $#1$}}     
\newcommand{\itvec}[1]{\mbox {\boldmath $#1$}}
\newcommand{\itvecsc}[1]{\mbox {\boldmath $\scriptstyle #1$}}
\newcommand{\gvec}[1]{\mbox{\boldmath $#1$}}

\newcommand{\balpha}{\mbox {\boldmath $\alpha$}}
\newcommand{\bbeta}{\mbox {\boldmath $\beta$}}
\newcommand{\bgamma}{\mbox {\boldmath $\gamma$}}
\newcommand{\bdelta}{\mbox {\boldmath $\delta$}}
\newcommand{\bepsilon}{\mbox {\boldmath $\epsilon$}}
\newcommand{\bvarepsilon}{\mbox {\boldmath $\varepsilon$}}
\newcommand{\bzeta}{\mbox {\boldmath $\zeta$}}
\newcommand{\boldeta}{\mbox {\boldmath $\eta$}}
\newcommand{\btheta}{\mbox {\boldmath $\theta$}}
\newcommand{\bvartheta}{\mbox {\boldmath $\vartheta$}}
\newcommand{\biota}{\mbox {\boldmath $\iota$}}
\newcommand{\blambda}{\mbox {\boldmath $\lambda$}}
\newcommand{\bmu}{\mbox {\boldmath $\mu$}}
\newcommand{\bnu}{\mbox {\boldmath $\nu$}}
\newcommand{\bxi}{\mbox {\boldmath $\xi$}}
\newcommand{\bpi}{\mbox {\boldmath $\pi$}}
\newcommand{\bvarpi}{\mbox {\boldmath $\varpi$}}
\newcommand{\brho}{\mbox {\boldmath $\rho$}}
\newcommand{\bvarrho}{\mbox {\boldmath $\varrho$}}
\newcommand{\bsigma}{\mbox {\boldmath $\sigma$}}
\newcommand{\bvarsigma}{\mbox {\boldmath $\varsigma$}}
\newcommand{\btau}{\mbox {\boldmath $\tau$}}
\newcommand{\bupsilon}{\mbox {\boldmath $\upsilon$}}
\newcommand{\bphi}{\mbox {\boldmath $\phi$}}
\newcommand{\bvarphi}{\mbox {\boldmath $\varphi$}}
\newcommand{\bchi}{\mbox {\boldmath $\chi$}}
\newcommand{\bpsi}{\mbox {\boldmath $\psi$}}
\newcommand{\bomega}{\mbox {\boldmath $\omega$}}

\newcommand{\bolda}{\mbox {\boldmath $a$}}
\newcommand{\bb}{\mbox {\boldmath $b$}}
\newcommand{\bc}{\mbox {\boldmath $c$}}
\newcommand{\bd}{\mbox {\boldmath $d$}}
\newcommand{\bolde}{\mbox {\boldmath $e$}}
\newcommand{\boldf}{\mbox {\boldmath $f$}}
\newcommand{\bg}{\mbox {\boldmath $g$}}
\newcommand{\bh}{\mbox {\boldmath $h$}}
\newcommand{\bp}{\mbox {\boldmath $p$}}
\newcommand{\bq}{\mbox {\boldmath $q$}}
\newcommand{\br}{\mbox {\boldmath $r$}}
\newcommand{\bs}{\mbox {\boldmath $s$}}
\newcommand{\bt}{\mbox {\boldmath $t$}}
\newcommand{\bu}{\mbox {\boldmath $u$}}
\newcommand{\bv}{\mbox {\boldmath $v$}}
\newcommand{\bw}{\mbox {\boldmath $w$}}
\newcommand{\bx}{\mbox {\boldmath $x$}}
\newcommand{\by}{\mbox {\boldmath $y$}}
\newcommand{\bz}{\mbox {\boldmath $z$}}



\title{\LARGE \bf  
 Distributed Resource Allocation Using One-Way Communication with Applications to Power Networks
}
\author{Sindri Magn\'{u}sson, Chinwendu Enyioha, Kathryn Heal, Na Li, Carlo Fischione, and Vahid Tarokh
\thanks{This work was supported by the VR Chromos Project and NSF grant No. 1548204 and NSF CAREER 1553407.}
\thanks{C. Enyioha, K. Heal, N. Li, and V. Tarokh are with the School of Engineering and Applied Sciences, Harvard University, Cambridge, MA~USA. (email: cenyioha@seas.harvard.edu; kathrynheal@g.harvard.edu; nali@seas.harvard.edu; vahid@seas.harvar.edu) }%
\thanks{S. Magn\'{u}sson and C.~Fischione are with Electrical Engineering School, Access Linnaeus Center, KTH Royal Institute of Technology, Stockholm, Sweden. (e-mail:{ sindrim@kth.se; carlofi@kth.se})}

 } 
\maketitle

 \begin{abstract}

Typical coordination schemes for future power grids require two-way communications. Since the number of end power-consuming devices is large, the bandwidth requirements for such two-way communication schemes may be prohibitive.  Motivated by this observation, we study distributed coordination schemes that require only one-way limited communications. 
In particular, we investigate how dual descent distributed optimization algorithm can be employed in power networks using one-way communication. In this iterative algorithm, system coordinators broadcast coordinating (or pricing) signals to the users/devices who update power consumption based on the received signal. Then system coordinators update the coordinating signals based on the physical measurement of the aggregate power usage. We provide conditions to guarantee the feasibility of the aggregated power usage at each iteration so as to avoid blackout. Furthermore, we prove the convergence of algorithms under these conditions, and establish its rate of convergence. We illustrate the performance of our algorithms using numerical simulations.  These results show that one-way limited communication may be viable for coordinating/operating the future smart grids.

 
 \end{abstract}

  \section{Introduction}
  
 Network infrastructures have a central role in communication, finance, technology and other areas of the economy, for instance, electricity transmission and distribution. Due to the massive scale of real-world networks, distributed optimization algorithms are expected to increasingly play a leading role in their operation. Distributed optimization problems are commonly solved by dual-based (sub)gradient methods where primal and dual decision variables are optimized at each iteration with some sort of message passing protocol between agents in the network \cite{kelly1998rate,nedic2008subgradient,low1999optimization}. In power networks, for example, consumers and suppliers communicate power consumption request and pricing signals back and forth.  This requires ubiquitous, two-way communication among consumers and suppliers. \textit{Since the number of end power-consuming devices is large, the bandwidth requirements for two-way schemes may be prohibitive}. However, the communication infrastructure for the grid, especially the power distribution networks, is still under-developed and the non-availability of communication bandwidth remains a challenge \cite{galli2011grid,biglieri2003coding,clarke1999internet,snow2001network,ericsson2010cyber,gungor2013survey}.
 
  Given communication bandwidth constraints, it is critical that algorithms for distributed coordination in power networks use significantly less communication overhead -- for instance, one-way message passing, as opposed to the conventional two-way message passing scheme typically used in distributed coordination algorithms. That said, a one-way message passing protocol for power distribution networks in which suppliers are only able to send out a price signal can be challenging to implement especially since a widespread system breakdown (blackouts) can occur if the aggregate power consumption exceeds the supply capacity.  Given \emph{1)} limits on power capacity at the supply end, and \emph{2)} a one-way message passing protocol, with many power demanding devices on the network, the question of interest is how to ensure that  the aggregate power consumption from users does not exceed the available supply capacity. In developing distributed algorithms to solve this problem, it is important that while the algorithm runs and users locally compute their optimal power allocation, the aggregate power consumption does not exceed the supply capacity to avoid a system failure and catastrophic blackout. 
 

Our work relates closely to \cite{beck2014gradient} where a distributed resource allocation algorithm was presented based on a fast dual gradient method. It showed that dual gradient iteration has a convergence rate of $\mathcal{O}(1/t^2)$ and yet the primal decision variables converge with a rate of $\mathcal{O}(1/t)$ (where $t$ denotes a time-step of the algorithm). The framework presented there, however, gives no guarantee of primal feasibility at each iteration of the algorithm.
The focus of this paper is on developing a distributed power allocation algorithm in which one-way communication, from the power suppliers to users, is used to coordinate the allocation. Due to the critical nature of power distribution systems, the power allocation algorithm presented in this paper
 \begin{itemize}
 \item maintains primal feasibility at each iteration and satisfies the capacity constraint, to avoid a blackout,
 \item assumes a one-way communication or information sharing model, in which the supplier (or network operator) sends out a \textit{price} signal to users, users update the consumption based on the signal but do not report the consumption back to the supplier, and
 \item has fast convergence rates and properties.
 \end{itemize}

 \subsection{Contributions of This Work}
We present a distributed, dual gradient power allocation algorithm for power distribution networks. The algorithm uses only one way communication, where the supplier iteratively broadcasts a price/dual variable to the users and then measures the aggregate power usage to compute the dual gradient. Since the users actually consume power during operation (\textit{en route} convergence) of the algorithm, it essential that the aggregate power usage does not exceed the supplier’s capacity to avoid a blackout.  We show how such blackouts can be avoided by providing step-sizes that ensure that the primal iterates remain feasible at ever iteration. To ensure the primal feasibility we must sacrifice the optimal convergence rate $O(1/t^2)$ \cite{beck2014gradient} for gradient methods for convex problems with Lipschitz continuous gradients  and instead only get the convergence rate $O(1/t)$. Nevertheless, we prove that under mild conditions on the problem structure our algorithm attains a convergence rate $O(c^t)$, where $0<c<1$ and $t$ is the time-step of the algorithm in the dual variable, and still guarantee primal feasibility.  Moreover, we provide conditions which ensure that the linear convergence rate, i.e., the constant $c$, does not depend on the number of users, which demonstrate excellent scaling properties of the algorithm.
  
The rest of the paper is organized as follows: Following notation and definitions, we introduce the system model, underlying assumptions and distributed solution in Section \ref{sec: SM SDP}. Sections \ref{sec:CA_GCR} and \ref{sec:CA_LCR} respectively present a convergence analysis and convergent rate analysis of our proposed algorithm. In Section \ref{sec:Numerical}, illustration of our algorithm  is presented and discussed. We make our conclusions in Section \ref{sec:Conclude}.

 \subsection{Notation and Definitions}

\begin{defin}
  We say that a function $f:\R^n \rightarrow \R$ is $L$-smooth on $\mathcal{X}\subseteq \R^n$ if its gradient is $L$-Lipschitz continuous on $\mathcal{X}$, i.e., for all $\vec{x}_1,\vec{x}_2\in \mathcal{X}$ we have
  \begin{align}
     || \nabla f(\vec{x}_1)-\nabla f(\vec{x}_2) || \leq L || \vec{x}_1-\vec{x}_2||.
  \end{align}
\end{defin}
\begin{defin}
  We say that a function $f:\R^n \rightarrow \R$ is $\mu$-convex on $\mathcal{X}\subseteq \R^n$ if for all $\vec{x}_1,\vec{x}_2\in \mathcal{X}$,
  \begin{multline}
   f(\vec{x}_2) \geq f(\vec{x}_1) +\langle \nabla f(\vec{x}_1), \vec{x}_2-\vec{x}_1\rangle + \frac{\mu}{2} || x_2-x_1 ||^2.
  \end{multline}
  Similarly, we say a function is $\mu$-concave if it is $\mu$-strongly concave.
\end{defin}


 \section{System Model and Algorithm}
 \label{sec: SM SDP}
In this paper we focus on an abstract power distribution model, and for simple exposition we ignore some specific power network constraints. We consider an electric network with $N$ users whose set is denoted by $\mathcal{N}$, and a single supplier. 
 The power allocation of user $i\in \mathcal{N}$ is denoted by $q_i\in \mathbb{R}_+ $ and the total power capacity is denoted by $Q\in \mathbb{R}_+$. 
 The value of the power $q_i$ of each user $i\in \mathcal{N}$ is decided by a \emph{private} utility function $U_i(q_i)$, which the user would like to maximize. By private, we mean the function $U_i(q_i)$ of each user is unknown to the power supplier.
 The \emph{Resource Allocation} problem for the power distribution system is given by the following optimization program~\cite{low1999optimization} \cite{liu2013optimal}:
\begin{empheq}[box=\fbox]{align}
   &\underline{\text{Resource Allocation (RA):}} \notag \\
   \tag{RA}  \label{eq:main_problem} 
  & \begin{aligned} 
     & \underset{q_1,\cdots,q_N}{\text{maximize}} 
     & & \sum_{i=1}^N U_i(q_i),  \\ 
     & \text{subject to}  
     & & \sum_{i=1}^N q_i \leq Q, \\ 
     &&&   m_i \leq  q_i \leq M_i. 
   \end{aligned} 
\end{empheq}
 We make the following assumptions on~\eqref{eq:main_problem}:
 \begin{assumption}(Convexity) \label{assump: good utilities}
    The function $U_i(q_i)$ is $\mu$-concave and increasing for all $i\in \mathcal{N}$.
 \end{assumption}
 \begin{assumption}(Well Posed)    \label{assump: not a stupid problem}
    We have that
    \begin{align} \label{eq:not a stupit problem}
        \sum_{i=1}^N m_i \leq Q \leq \sum_{i=1}^N M_i.
    \end{align}
    In other words, Problem~\eqref{eq:main_problem} is feasible and the constraint $\sum_{i=1}^N q_i\leq Q$ is not redundant.
 \end{assumption}

 As motivated in the introduction, in this work we focus on one-way communication protocols for solving~\eqref{eq:main_problem}. For each user $i\in \mathcal{N}$ the problem data $U_i$, $q_i$, $m_i$, and $M_i$ are private to user $i$.  In particular, the supplier can not access this information. We consider protocols based on the following two operations.

  \begin{operation}[One-Way Communication] \label{assumption:one-way-comm}
     At each time-step (iteration), the supplier can broadcast to all users one scalar message referred to as the price (of power). 
  
  \end{operation}
  \begin{operation}[Feedback Information] \label{assumption:data-separation}
       After broadcasting the price, the supplier  can measure the deviation between total power load and supply, i.e.   $  \sum_{i=1}^N q_i(t) - Q$, where $t$ indicates the time-step.
  \end{operation}
  Algorithms for solving~\eqref{eq:main_problem}  that only use  Operations~\ref{assumption:one-way-comm} and~\ref{assumption:data-separation} can be achieved using duality theory \cite{kelly1998rate}\cite{low1999optimization}. In particular, we achieve the solution by solving the dual problem of~\eqref{eq:main_problem} that is given as follows: 
\begin{empheq}[box=\fbox]{align}
   &\underline{\text{The Dual of (RA):}} \notag \\
   \tag{Dual-RA}   \label{eq:dual_problem}
  &    \begin{aligned}
     & \underset{p }{\text{minimize}}
     & &  D (p),  \\
     & \text{subject to} 
     & &  p \geq 0,
   \end{aligned}
\end{empheq}
  where $D$ and $p$ are the dual function and dual variables~\cite[Chapter 5]{convex_boyd}, respectively, and $D$ is given by
  \begin{align}
      D(p) =& \ \underset{ \vec{q}\in \mathcal{M}}{\text{maximize}}~   \sum_{i=1}^N U_i(q_i) {-} p \left(   \sum_{i=1}^N q_i {-} Q \right) \nonumber \\
             =& \ \sum_{i=1}^N U_i(q_i(p)) {-} p \left(   \sum_{i=1}^N q_i(p) {-} Q \right)  \label{eq:dual2} 
  \end{align}
  where $m_i$ and $M_i$ are respectively the lower and upper bounds on the power loads of user $i$; $\mathcal{M}= \prod_{i=1}^N [m_i,M_i]$, and $\vec{q}=(q_1,\hdots, q_N)^T$. The local problem for user $i$ is to solve 
  \begin{align} \label{eq:demand_from_price} 
     q_i (p) {=}  \underset{q_i \in [m_i,M_i]}{\text{argmax}} 
              U_i(q_i) - p ~ q_i {=} \left[(U_i')^{-1}(p) \right]_{m}^{M}.
  \end{align}
   In words, $q_i (p)$ in~\eqref{eq:demand_from_price} denotes the power demand of user $i$ when the price of power is $p$.
   We have the following relationship between~\eqref{eq:main_problem} and~\eqref{eq:dual_problem}:
  \begin{lemma}(Strong Duality)
       Suppose Assumptions~\ref{assump: good utilities} and~\ref{assump: not a stupid problem} hold;  
   if $p^{\star}$ is an optimal solution of~\eqref{eq:dual_problem}, then $\vec{q}(p^{\star})=(q_i(p^{\star})_{i\in \mathcal{N}}$ (cf.~\eqref{eq:demand_from_price}) is the optimal solution to~\eqref{eq:main_problem}.
  \end{lemma}
  \begin{proof}
     The proof follows by simply noting that $\sum_{i=1}^N m_i \leq Q$, and $m_i \leq M_i$ ensures that~\eqref{eq:main_problem}  satisfies Slater's condition, which is sufficient for the zero duality gap for convex problems~\cite[Section 5.2.3]{convex_boyd}.
  \end{proof}
 
   Let us now demonstrate benefits of studying the dual problem~\eqref{eq:dual_problem}, which takes advantage of the following property:
    \begin{prop} \label{prop:dual-problem-is-1L-convex1}
     Suppose Assumption~\ref{assump: good utilities} holds, then $D(\cdot)$ is differentiable and $N/\mu$-smooth. 
  \end{prop}
  \begin{proof}
     See Lemma II.2 in~\cite{beck2014gradient}
  \end{proof}
 \noindent Due to Proposition~\ref{prop:dual-problem-is-1L-convex1}, we can solve~\eqref{eq:dual_problem} using a dual descent method, given by the iterates:
  \begin{align}
    p(t{+}1) = \lceil p(t) - \gamma D'(p(t)) \rceil^+,
  \end{align}
  where $\gamma>0$ is step-size and $D'(\cdot)$ is the gradient of $D(\cdot)$, which is given by
  \begin{align} 
   D'(p) = Q-\sum_{i=1} q_i(p). \label{eq:dual_gradient}
  \end{align}
  Notice that~\eqref{eq:dual_gradient} is exactly the feedback provide in Operation~\eqref{assumption:data-separation} above.
 In Algorithm~\ref{Alg:staticDD} below  we demonstrate how the the dual descent algorithm can be implemented between the supplier and the users using only Operations~\ref{assumption:one-way-comm} and~\ref{assumption:data-separation}.

\begin{algorithm}

  \SetKwInOut{Initialization}{Initialization}
  \SetKw{CC}{Supplier:}
  \SetKw{User}{User $i$: }
    \CC{Decides the initial prize $p(0)\geq 0$}\;
    \For{$t=0,1\cdots$}{
     \CC{Broadcast $p(t)$ to all users}\;
   \For{$i=1,\cdots, N$}{
    \User{Receives $p(t)$ }\;
    \User{$q_i(t)= \left[(U_i')^{-1}(p(t)) \right]_{m}^{M} $ }\;
  }
       \CC{Measures $D'(p(t)){=}Q{-}\sum_{i=1}^N q_i(t)$}\;
       \CC{Updates price: $\hspace{-0.11cm} p(\hspace{-0.02cm}t\hspace{-0.02cm}{+}\hspace{-0.02cm}1\hspace{-0.02cm}){=} \lceil p(t){+}\hspace{-0.02cm}\gamma D'\hspace{-0.05cm}(\hspace{-0.02cm}p(\hspace{-0.02cm}t\hspace{-0.02cm})\hspace{-0.02cm}) \rceil^+\hspace{-0.06cm}$};
  }
 \caption{Dual descent, using only Operations~\ref{assumption:one-way-comm} and~\ref{assumption:data-separation} and local computation at the users and the supplier.}
 \label{Alg:staticDD}
\end{algorithm}

 \begin{remark} \label{remark:blackout} 
  Due to the one-way communication, the primal iterates $q_i(t)$ are not communicated to the supplier. 
  Instead the users take the action $q_i(t)$ and the aggregate of their actions is measured by the supplier. 
  Therefore, it is essential that the primal problem~\eqref{eq:main_problem}  is feasible during every iteration of Algorithm~\ref{Alg:staticDD}, i.e. $\sum_{i=1}^N q_i(t)\leq Q$ for all $t\in \N$.
    Otherwise, a heightened demand of power in the network  can result in a system overload and eventual blackout.
     
 \end{remark}

In what follows, we study the convergence of Algorithm \ref{Alg:staticDD} and how to ensure that the primal problem is feasible during every iteration. 

 \section{Convergence analysis of Algorithm ~\ref{Alg:staticDD}}

In this section, we provide rules on choosing the initial price and the step-size to ensure that the primal variables are feasible throughout the algorithm (Section ~\ref{sec:CA_GCR} ). We also further prove that the algorithm has a convergence rate $c^t$, where $0<c<1$ and $t$ is the time-step of the algorithm, under certain conditions (Section ~\ref{sec:CA_LCR}).

 \subsection{General Convergence Result }  \label{sec:CA_GCR}

  The following result on the convergence of Algorithm~\ref{Alg:staticDD} is standard in the literature~\cite{low1999optimization}.  
  \begin{prop} \label{prop:lows_paper}
     If $\gamma \in ]0,2\mu/N[$ in Algorithm~\ref{Alg:staticDD}, then every limit point of $p(t)$ is a solution to~\eqref{eq:dual_problem} and $\lim_{t\rightarrow \infty}\vec{q}(t)=\vec{q}^{\star}$ where $\vec{q}^{\star}$  is the solution to~\eqref{eq:main_problem}.
  \end{prop}
 \begin{proof}
   Follows directly from Theorem 1 in~\cite{low1999optimization}.
 \end{proof}
  Proposition~\ref{prop:lows_paper} ensures that the primal sequence $\vec{q}(t)$ converges to $\vec{q}^{\star}$. 
 However, there is no guarantee that that the iterates $\vec{q}(t)$ are feasible to~\eqref{eq:main_problem} at each iteration except in the limit.
  Such infeasibilities can cause blackouts in the power network (cf. Remark ~\ref{remark:blackout}) and therefore it is essential to find conditions that ensure the feasibility of $\vec{q}(t)$ for all $t\in \N$.
   Such conditions are now established. 
 \begin{prop} \label{prop:feasible_prop}
  Let $\mathcal{P}^{\star}=[\un{p}^{\star},\ov{p}^{\star}]$ be the set of optimal solutions to~\eqref{eq:dual_problem}:
     If Assumptions~\ref{assump: good utilities} and~\ref{assump: not a stupid problem} hold and $\gamma \in ]0,\mu/N]$ in Algorithm~\ref{Alg:staticDD}, then the following hold.
   \begin{enumerate}[a)]
  \item
    If $p(0)\geq \ov{p}^{\star}$, then $\lim_{t\rightarrow\infty} p(t)=\ov{p}^{\star}$ and  the primal variables $\vec{q}(t)$ are feasible to~\eqref{eq:main_problem} at every iteration $t\in \mathcal{N}$.
  
  \item The convergence rate of the objective function values $D(p(t))$ is $\mathcal{O}(1/t)$.
         Moreover, the optimal convergence rate 
    \begin{align}
          D(p(t))-D(p^{\star})\leq \frac{2 ||p(0)-p^{\star} ||^2}{t+4},
    \end{align} 
   is achieved when $\gamma= \mu/N$.
   \end{enumerate}
 \end{prop}
  \begin{proof}  a)     
     From Proposition~\ref{prop:lows_paper} we know that every limit point of $p(t)$ is in $\mathcal{P}^{\star}$.
     Moreover, since $\vec{q}(p)$ is decreasing function (cf. Equation~\eqref{eq:demand_from_price}) and $\vec{q}(\ov{p}^{\star})$ is feasible point of~\eqref{eq:main_problem}, then $\vec{q}(p)$ is also a feasible point of~\eqref{eq:main_problem} for all $p\geq \ov{p}^{\star}$. 
     Therefore, if we can show that $p(t)\geq\ov{p}^{\star}$ for all $t\in\N$, then it holds that
 $\lim_{t\rightarrow\infty} p(t)=\ov{p}^{\star}$ and  $\vec{q}(p(t))$ is feasible to~\eqref{eq:main_problem} for all $t\in \N$.
    
   Let us now show that $\ov{p}^{\star}<p(t)$ for all $t\in \N$ by induction.
     By assumption $\ov{p}^{\star}<p(0)$. 
     Let us now suppose $\ov{p}^{\star} < p(t)$. 
    Then $D'( \ov{p}^{\star})=0$ implies that $D'(p(t))>0$, and the convexity of $D(\cdot)$ implies that $D'(\cdot)$ is increasing.
         Hence, we have $p(t+1)=p(t)-\gamma D'(p(t))<p(t)$.
     Moreover, 
    using that $D(\cdot)$ is $N/\mu$-smooth and $D'(\cdot)$ is increasing we get 
    \begin{align}
        D'(p(t)){-} D'(p(t{+}1)) {<} \frac{N}{\mu} (p(t) {-}p(t{+}1)) {=} \frac{N}{\mu} \gamma D(p(t)), \notag
     \end{align}
     and rearranging one obtain s
    \begin{align}
      0< \left(1-\frac{N}{\mu} \gamma\right)D'(p(t)) \leq  D'(p(t{+}1)).
     \end{align} 
      Hence, the step-size $\gamma\in]0,\mu/N[$ implies that $D'(p(t{+}1))>0$.
      Hence, since $D'(\cdot)$ is increasing and $D'(p^{\star})=0$ we must have $p^{\star}\leq p(t{+}1)$.
  
 b) This result is proved in~\cite[Theorem 2.1.14]{Book_Nesterov_2004}.
  \end{proof}

  \begin{remark} \label{remark:1t}
      For general objective functions with Lipschitz continuous gradients, the optimal convergence rate is $\mathcal{O}(1/t^2)$ by using Nesterov fast gradient methods, see~\cite[chapter~2]{Book_Nesterov_2004}.
      However, these optimal dual gradient methods cannot ensure the feasibility of the primal problems during the converging process, bringing blackout risk. 
  \end{remark}
 
 \subsection{Linear  Convergence Rate} \label{sec:CA_LCR} 

 We now identify structures on~\eqref{eq:main_problem} which ensure convergence rate $c^t$, where $0<c<1$, on Algorithm~\ref{Alg:staticDD}. 
   We start by providing a linear rate under the following assumptions on utility function and local constraints.
   
   \begin{assumption} \label{assump:smooth}
 The function $U_i(q_i)$ is $L$-Lipschitz continuous.   
   \end{assumption}
 \begin{assumption}  \label{assump: more realistic case}
    Let $\un{p}_i=U_i'(M_i)$ and $\ov{p}_i=U_i'(m_i)$.
    Then the set $\mathcal{P}= \bigcup_{i\in \mathcal{N}} [ \un{p}_i,\ov{p}_i]$ is connected, i.e., an interval on $\R_+$. 
     We write $\mathcal{P}$ in terms of its end points as $\mathcal{P}= [\un{p},\ov{p}]$. 
 \end{assumption}
  The linear convergence rate can now formally be stated as follows.
   \begin{prop} \label{prop:GeneralLinearConvergence}
Suppose Assumptions~\ref{assump: good utilities} to \ref{assump: more realistic case} hold; then for any $p(0)\in [\un{p},\ov{p}]$ and $\gamma \in ]0, \mu/N[$, Algorithm~\ref{Alg:staticDD} converges at rate $c^t$ to the set of optimal solutions $\mathcal{P}^{\star}$; and 
       \begin{align}
            \dist(p(t),\mathcal{P}^{\star})  \leq 
             &   c^t \dist(p(0),\mathcal{P}^{\star}),  \label{prop:eq_linear convergece rate1}
       \end{align}
     where $c=1-\gamma/L$. The optimal convergence rate $c=1-\mu/(N L)$ is obtained when $\gamma=\mu/N$.
   \end{prop}
  \begin{proof}     Let us start by showing that $D(\cdot)$ is $1/L$-convex on $[\un{p},\ov{p}]$.
     From Lemma~\ref{prop:dual-problem-is-1L-convex}, we can write $D(p)=\sum_{i=1}^N D_i(p)$ where $D_i(\cdot)$ are convex on $\R_+$ for all $i\in \mathcal{N}$ and $D_i(\cdot)$ is $1/L$-convex on $[\un{p}_i,\ov{p}_i]$.
      Therefore, from Assumption~\ref{assump: more realistic case} and the fact that sums of a convex and $1/L$-convex function are $1/L$-convex~\cite[Lemma 2.1.4]{Book_Nesterov_2004} we get that $D(p)$ is $1/L$ convex on $[\un{p},\ov{p}]$.

     Since $D(\cdot)$ is strongly convex on $\mathcal{P}$, $\mathcal{P}^{\star}\cap [\un{p},\ov{p}]$ can have at most one element.
     In fact, $\mathcal{P}^{\star}\cap [\un{p},\ov{p}]$ is non-empty, since from~\eqref{eq:dual_gradient} and Assumption~\ref{assump: not a stupid problem} we have
    \begin{align} 
       D'(\un{p})&=Q-\sum_{i=1}^{N}M_i \leq 0 ~~~ \text{ and  }  \\ D'(\ov{p})&=Q-\sum_{i=1}^{N}m_i \geq 0,
    \end{align}
    and therefore by the continuity of $D'(\cdot)$ and the intermediate value theorem there exists  $p^{\star}\in[\un{p},\ov{p}]$ such that $D'(p^{\star})=0$.
      We now show the $c^t$ convergence rate of Algorithm~\ref{Alg:staticDD} to $p^{\star}$. 
      Without loss of generality suppose that  $p(0)\geq p^{\star}$.
      Then from Proposition~\ref{prop:feasible_prop} and the step-size choice of $\gamma\in ]0,\mu/N]$ we know $p(t)\geq  p^{\star}$ for all $t\in \N$.
   Hence, suppose $p(0)\geq p^{\star}$ then we get that
   \begin{align}
      p(t{+}1)-p^{\star} &=   p(t)-p^{\star} - \gamma D'(p(t)),  \nonumber \\
                                 &\leq p(t)-p^{\star}-\frac{\gamma}{L}(p(t)-p^{\star}) \\
                                 & =  \left(1-\frac{\gamma}{L}\right) (p(t)-p^{\star}) \label{prop:lincon_c}
   \end{align}
   where the inequality follows from the fact that $p(t{+}1),p(t),p^{\star}\in [\un{p},\ov{p}]$ and that $D$ is $1/L$-strongly convex on $[\un{p},\ov{p}]$, which implies that~\cite[Theorem 2.1.9.]{Book_Nesterov_2004} 
  \begin{align}
        \frac{1}{L}(p-p^{\star}) \leq D'(p), \text{ for all } p \in   [\un{p},\ov{p}].
  \end{align}
   Applying Inequality~\eqref{prop:lincon_c} $t$ times gives~\eqref{prop:eq_linear convergece rate1}, which concludes the proof.
  \end{proof}
   Note that the optimal convergence rate $c=1- \mu/(N L)$ in \eqref{prop:eq_linear convergece rate1} depends on the number of users $N$.
   However, the following mild assumption yields a convergence rate in the dual variable that is independent of the number of users, which illustrate excellent scaling properties Algorithm \ref{Alg:staticDD}.

   \begin{assumption}  \label{assump: ideal  case}
    Let $\un{p}=U_i'(M_i)$ and $\ov{p}=U_i'(m_i)$ for all $i\in \mathcal{N}$.
 \end{assumption}

   Assumption~\ref{assump: ideal  case} essentially gives each user the freedom to choose their utility functions $U_i$ as well as upper and lower bounds $m_i$ and $M_i$ with the constraint that $U_i'(M_i)=\un{p}$ and $U_i'(m_i)=\ov{p}$. These utility functions can be customized to reflect the preferences or priorities of the user, for example the times of day that each user will require the most power.

   \begin{prop}\label{prop:convergence rate}
       Suppose Assumptions~\ref{assump: good utilities}--\ref{assump:smooth} and~\ref{assump: ideal  case} hold.
       Then for any $p(0)\in [\un{p},\ov{p}]$ and $\gamma \in]0, \mu/N]$, Algorithm~\ref{Alg:staticDD} has a $c^t$ convergence rate which does not depend on the number of users.
     In particular, we have
       \begin{align}
            \dist(p(t),\mathcal{P}^{\star})  \leq 
             &   c^t \dist(p(0),\mathcal{P}^{\star}),  \label{prop:eq_linear convergece rate}
       \end{align}
     where $c=1-N\gamma/L$. The optimal convergence rate $c=1- \mu/L$ is obtained when $\gamma=\mu/N$.
   \end{prop}
    \begin{proof}
   The proof follows almost the  same steps as the proof of Proposition~\ref{prop:GeneralLinearConvergence}.
  The only difference is that here  the dual function $D$ is $N/L$-convex on $[\un{p},\ov{p}]$ instead of $1/L$-convex as in Proposition~\ref{prop:GeneralLinearConvergence}.
   To see why $D$ is $N/L$-convex on $[\un{p},\ov{p}]$, we have from Lemma~\ref{prop:dual-problem-is-1L-convex} (in the Appendix)
   that $D(p)=\sum_{i=1}^N D_i(p)$ where $D_i(p)$ is $1/L$-convex on $[\un{p},\ov{p}]$ for each $i\in \mathcal{N}$.
    \end{proof}
    
    In the next section we will consider a specific form of utility function that is commonly used in the resource allocation literature \cite{yim2007reverse}. In particular, we assume a utility function of the form
   \begin{align}\label{eq:NumResCase2}
       U_i(q_i)=a_i \log(b_i+q_i),
    \end{align}
where the parameters $a_i,b_i$ may be unique to the different users $i\in\mathcal{N}$. Applying Proposition 4 to this logarithmic utility function, we can assert specific conditions under which the set $\mathcal{P}= \bigcup_{i\in \mathcal{N}} [ \un{p}_i,\ov{p}_i]$ is connected. 
   \begin{prop}\label{prop:inNumCase2}
       Let $a_i$ in Equation~\eqref{eq:NumResCase2} be ordered as $a_1\leq \cdots \leq a_N$. Let $M_i=M$ and $m_i=0$ for all $i$.
       Then if  $(b_i+M)a_i\geq a_{i+1}$ for $i=1,\cdots, N-1$ and $p(0)\in [a_1/(b_i+M),a_N]$ then if we choose the step-size $\gamma=\mu/N$ we have
       \begin{align}
          \dist(\vec{p}(t),\mathcal{P}^{\star}) \leq \left( 1-\frac{\mu}{N L}\right)^t  \dist(\vec{p}(0),\mathcal{P}^{\star})
       \end{align}  

   \end{prop}
   \begin{proof} 
      The result is a corollary of Proposition~\ref{prop:GeneralLinearConvergence} and can be obtained by showing that Assumption~\ref{assump: more realistic case} holds.

  To see that Assumption~\ref{assump: more realistic case} holds, note that  $\un{p}_i=U_i'(M)= a_i/(b_i+M)$ and  $\ov{p}_i=U_i'(m)=a_i$.
        Therefore, $(b_i+M)a_i\geq a_{i+1}$ for $i=1,\cdots, N-1$ implies that $\mathcal{P}= \bigcup_{i\in \mathcal{N}} [ \un{p}_i,\ov{p}_i]$ is connected.
   \end{proof}
   \begin{figure}
   \includegraphics[width=\linewidth]{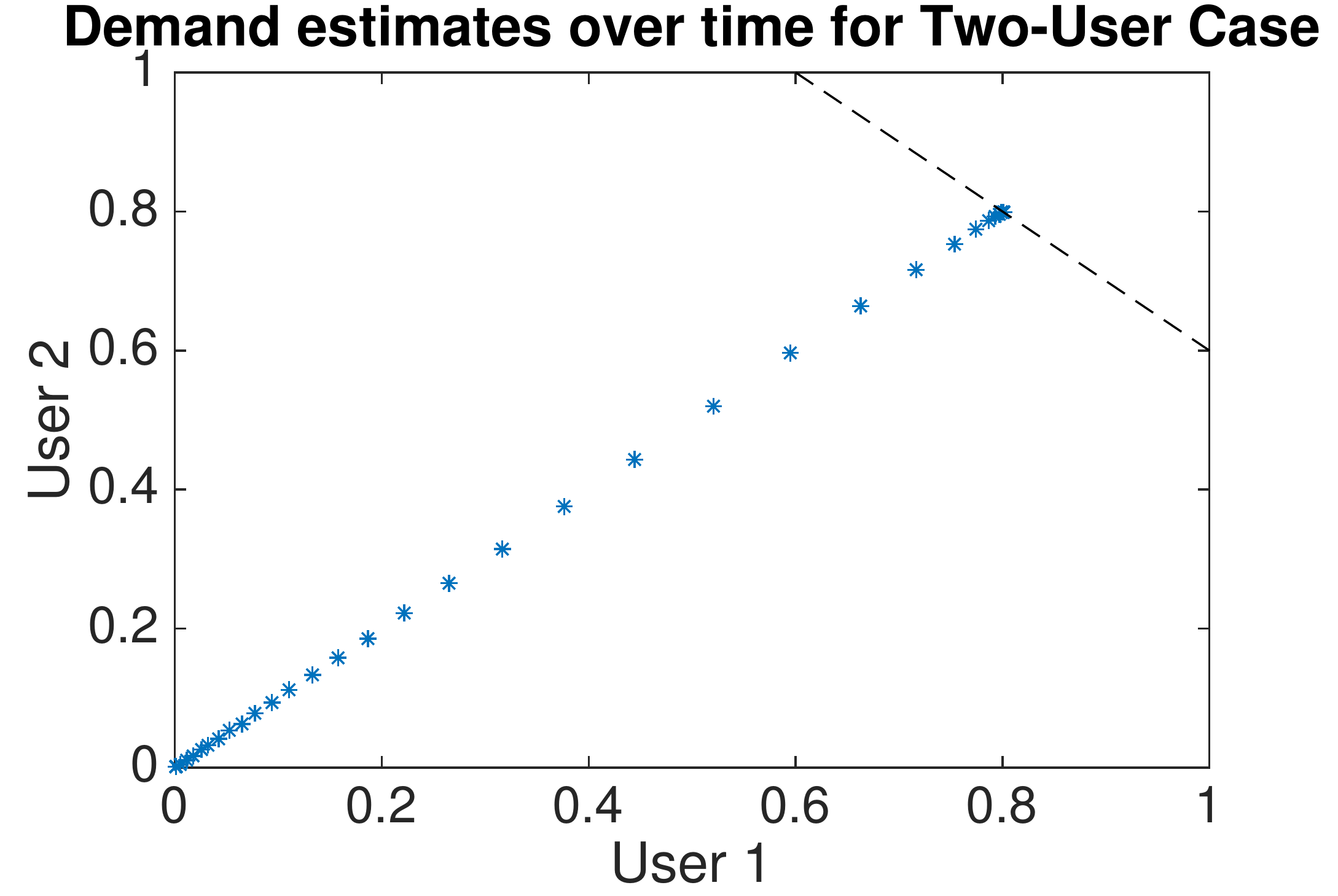}
 \caption{Using a sufficiently small step size, the feasibility of the primal problem is maintained. The upper boundary of this feasible set is denoted by the dotted line.}\label{fig:rate per user 2-case}
   \end{figure}

\begin{figure*}
    \centering
    \begin{subfigure}[b]{0.32\textwidth}
        \includegraphics[width=\textwidth]{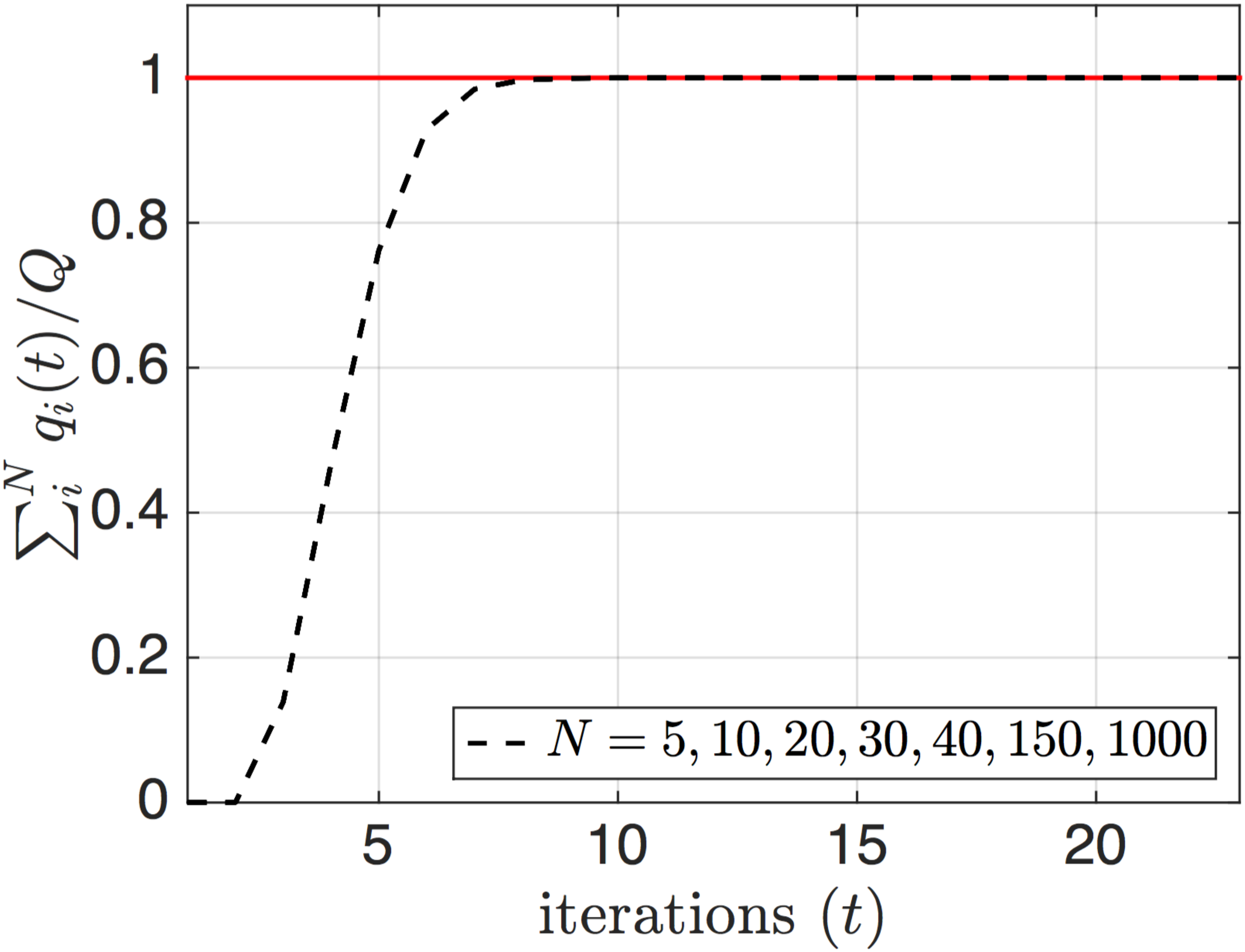}
        \caption{Primal Feasibility}
       \label{fig:primal feasibility-high price}
    \end{subfigure}
    ~ 
    \begin{subfigure}[b]{0.32\textwidth}
        \includegraphics[width=\textwidth]{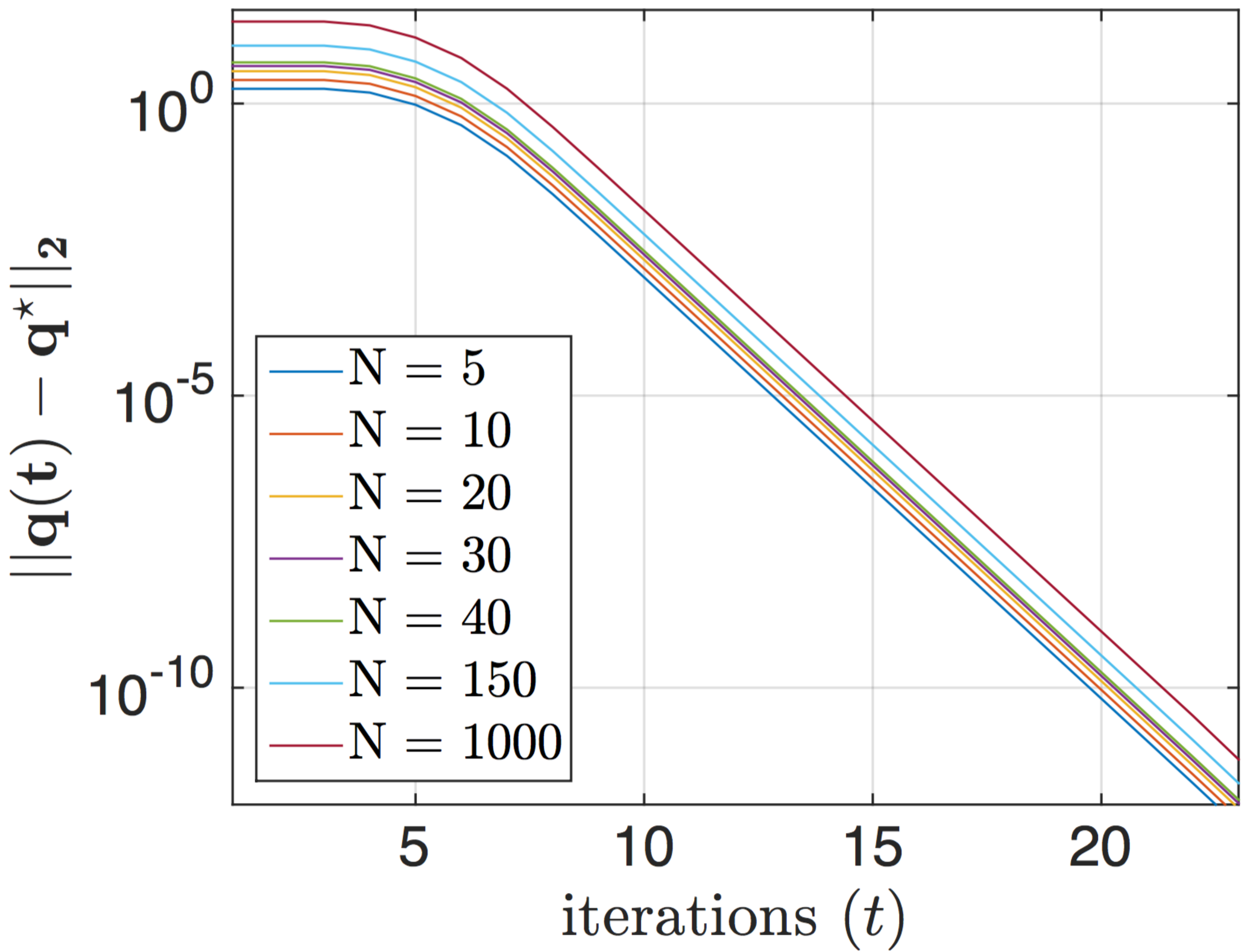}
        \caption{Convergence of the Primal Variable}
        \label{fig:primal-error}
    \end{subfigure}
    ~ 
    \begin{subfigure}[b]{0.32\textwidth}
        \includegraphics[width=\textwidth]{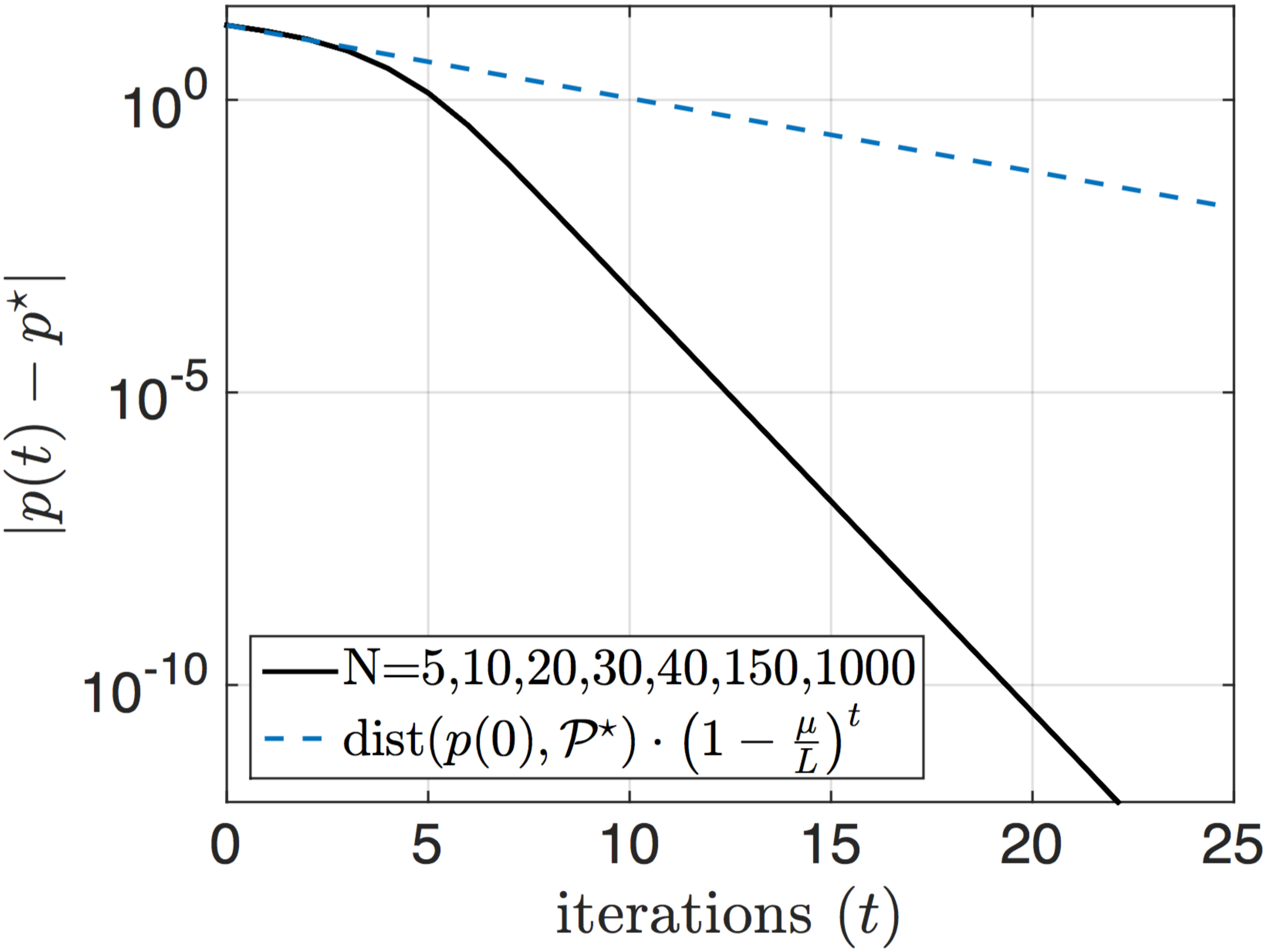}
        \caption{Convergence of the dual Variable}
       \label{fig:dual-error}
    \end{subfigure}
    \caption{Convergence of the algorithm for different number of users $N=5,10,20,30,40,150,1000$. In Figures~\ref{fig:primal feasibility-high price} and \ref{fig:dual-error} the behaviour of the algorithm is identical for all $N$. The blue dotted line in~\ref{fig:dual-error} is the theoretical bound from Proposition~\ref{prop:convergence rate}.  }\label{fig:multiuser case}
\end{figure*}

 \section{Numerical Results}\label{sec:Numerical}

 In this section we illustrate the theoretical findings using two numerical examples.  
 In both examples, we set $m_i=0$ and $M_i$ and use the same utility function of the form in Equation~\eqref{eq:NumResCase2} with $a_i=20$, $b_i=1$.


 Figure~\ref{fig:rate per user 2-case} depicts the first example which is a system consisting of two users. 
 The power suppliers capacity is $Q=1.6$ and the  step size is $\gamma=\mu/N$. 
 As observed, both users are assigned a socially optimal rate of $q_i = 0.8$ power units each and
  primal feasibility  is satisfied at each iteration of the algorithm, i.e., $\sum_{i=1}^N q_i(t) \leq Q $ for all~$t$. 


 Figure~\ref{fig:multiuser case} illustrate a scenario with multiple users $N=5,10,20,30,40,150,1000$.
 The step-size is $\gamma=\mu/N$, the initial price is $p(0)=30$ in all cases, and the suppliers capacity is $Q=4N/5$. 
 The assumptions of Proposition~\ref{prop:convergence rate} hold in this case.
  Figure \ref{fig:primal feasibility-high price} depicts the  aggregate rates $\sum_{i=1}^N q_i$ allocated to all users  scaled by the power capacity $Q$. 
 As observed, the primal variables remain feasible at each iteration of the algorithm, as proved in Proposition~\ref{prop:feasible_prop}.
Figures \ref{fig:primal-error} and \ref{fig:dual-error} respectively illustrate the convergence of the algorithm to the optimal primal and dual variables.
 As observed,  the dual variables converge at a linear rate,
 which is to be expected due to Proposition~\ref{prop:convergence rate}.
 Moreover, the convergence is not sensitive to the number of users $N$ which also agrees with the results of  Proposition~\ref{prop:convergence rate}.
  Interestingly, the primal variables also converge at a linear rate which suggest that similar convergence results can be obtained for the primal variables. 
 
 As can be observed, the results in figures~\ref{fig:primal feasibility-high price} and~\ref{fig:dual-error} are identical for  $N=5,10,20,30,40,150,1000$. This can be explained by that the utilities of of the users are identical and the power supply increases with number of users.  However, Proposition~\ref{prop:convergence rate} also holds under much more general assumptions.  In future work we will consider more  complicated and realistic problems for which Proposition~\ref{prop:convergence rate} holds.

\section{Summary and Conclusion}
\label{sec:Conclude}
We considered the problem of allocating power to users in a network by one supplier and presented a distributed algorithm that uses a one-way communication between the supplier and users to coordinate the optimal solution. We showed how to choose step sizes and derive appealing convergence properties in the dual domain. Furthermore, we identified mild modifications to the power allocation problem that \textit{1)} results in a $c^t$ convergence rate of our algorithm, and \textit{2)} yields a convergence rate independent of the number of users in the network. All these indicate excellent scaling properties of the algorithm. The results presented here propose a paradigm for solving distributed resource
 allocation problems that achieve fast convergence rates, even with a one-way information passing mechanism. Based on this work, of interest to us, is to derive similar convergence rates in the primal domain. In addition, within the context of limited communication, we would like to study the tradeoffs between the need for primal feasibility, convergence rates and actual bits transmitted (bandwidth) using the one-way  coordination protocol.
 
\appendix
 The following Lemmas are used in our derivations.
  \begin{lemma} \label{prop:dual-problem-is-1L-convex}
     Suppose Assumption~\ref{assump: good utilities} holds and set $\un{p}_i=U_i'(M_i)$ and $\ov{p}_i=U_i'(m_i)$ for $i\in \mathcal{N}$.
     Then the dual function $D(\cdot)$ can be decomposed as $D(p)=\sum_{i=1}^N D_i(p)$, where $D_i(\cdot)$ is 
     $1/L$-convex on $[\un{p}_i,\ov{p}_i]$ and the derivative $D'(\cdot)$ is constant on $[0,\un{p}_i]$ and on $[\ov{p}_i,\infty]$. 
    In particular, $D_i'(p) =Q/N-M_i$ if $p\in [0,\un{p}_i]$ and  $D_i'(p) =Q/N-m_i$ if $p\in [\ov{p}_i,\infty]$.
  \end{lemma}
  \begin{proof}
  a) We set 
     \begin{align}\label{eq:decomposed dual fun}
        D_i(p) = U_i(q_i(p))-pq_i(p)- p \frac{Q}{N},
     \end{align}
     where $q_i(p)$ is defined as in~\eqref{eq:demand_from_price}.
     By summing~\eqref{eq:decomposed dual fun} over $i\in \mathcal{N}$ and recalling the definition of the dual function $D$ in~\eqref{eq:dual2}, we get
     \begin{align}
         D(p) = \sum_{i=1}^N D_i(p).
     \end{align}
     Let us now show that $D_i(\cdot)$ is $\frac{1}{L}$-convex on $[\un{p}_i,\ov{p}_i]$.
     From~\cite[Theorem~2.1.9.]{Book_Nesterov_2004} 
it     is  $\frac{1}{L}$-convex on $[\un{p}_i,\ov{p}_i]$ if and only if for all $y_1,y_2 \in [\un{p}_i(t),\ov{p}_i(t)]$ it holds that
   \begin{multline}
            \frac{1}{L}   | y_2{-}y_1|^2 \leq    \left\langle y_2- y_1 ,D_i'(y_2)- D_i'(y_1) \right\rangle \\ 
  ~=\left\langle y_2{-} y_1 ,\left({-}(U_i')^{-1}(y_2)\right){-} \left({-}(U_i')^{-1}(y_1) \right) \right\rangle,  \label{eq:lemma2sdfa}
    \end{multline}
     where the equality comes by using that  (cf.~\eqref{eq:dual_gradient} and~\eqref{eq:demand_from_price})
   \begin{align}
         D_i'(p)=Q/N-q_i(p)= Q/N-[(U_i')^{-1}(p)]_{m_i}^{M_i}. \label{eq:decomposed dual gradient}  
   \end{align}
  Due to the strong concavity of $U_i(\cdot)$ on $[m_i,M_i]$ (cf.~Assumption~\ref{assump: good utilities}),  $U_i'$ is bijective from $[m_i,M_i]$ to $[\un{p}_i,\ov{p}_i]$ and we can choose  $x_1,x_2 \in [m_i,M_i]$ such that $U_i'(x_1)=y_1$ and $U_i'(x_2)=y_2$.      
    Then since $U_i$ is concave and $L$-smooth we have from~\eqref{eq:concave-lemma_3} in Lemma~\ref{lemma:concave_and_L-smooth_iif} that
   \begin{align}
               {-}  \left\langle U_i'(x_2){-} U_i'(x_1) ,x_2 {-}x_1 \right\rangle 
      { \geq   }\frac{1}{L}   | U_i'(x_2){-}U_i'(x_1)|^2,  \notag 
    \end{align}
    or by using $U'(x_1)=y_1$ and $U'(x_1)=y_1$, we get
   \begin{align}
     \hspace{-0.1cm}  \left\langle y_2{-} y_1 ,\left({-}(U_i')^{-1}(y_2)\right){-} \left({-}(U_i')^{-1}(y_1) \right) \right\rangle 
      { \geq  } \frac{1}{L}   | y_2{-}y_1|^2.  \notag 
    \end{align}
    Hence,~\eqref{eq:lemma2sdfa} holds and we can conclude that $D_i(\cdot)$ is $1/L$-convex on $[\un{p}_i,\ov{p}_i]$.    Let us now show that the gradient $D_i'(\cdot)$ is constant on $[0,\un{p}_i]$ and on $[\ov{p}_i,\infty]$.
   The result follows by combining~\eqref{eq:decomposed dual gradient} with $U_i'$ is decreasing and that $U'(M_i)=\un{p}_i$ and $U'(m_i)=\un{p}_i$.
  \end{proof}

 \begin{lemma} \label{lemma:concave_and_L-smooth_iif}
    A function $f$ is concave and $L$-smooth on $\mathcal{X}$ if for all $\vec{x}_1,\vec{x}_2 \in \mathcal{X}\subseteq \R^n$,
   \begin{eqnarray}
         \left\langle f'(\vec{x}_1){-} f'(\vec{x}_2) ,\vec{x}_2 {-}\vec{x}_1 \right\rangle 
      { \geq}   \frac{1}{L}   || f'(\vec{x}_2){-}f'(\vec{x}_1)||^2,   \label{eq:concave-lemma_3} 
    \end{eqnarray}    
    
 \end{lemma}
\begin{proof}
The proof follows directly from Theorem 2.1.5. in \cite{Book_Nesterov_2004} and using that $-f$ is convex. 
\end{proof}  

\bibliographystyle{IEEEtran}
\bibliography{IEEEabrv,refs}

\end{document}